\documentclass[12pt]{amsart}
	
	\usepackage{amssymb,amsmath,amscd,amsthm}
	\usepackage{graphicx,psfrag,epsfig}
	\usepackage[active]{srcltx}
	\usepackage{thm-restate}
	\usepackage{color,soul}
	\usepackage{xcolor}

	\newtheorem{theorem}{Theorem}[section]

	\newtheorem{conjecture}[theorem]{Conjecture}
	\newtheorem{lemma}[theorem]{Lemma}
	\newtheorem{proposition}[theorem]{Proposition}
	\newtheorem{defn}[theorem]{Definition}

	\newcommand{\U}{\begin{pmatrix}U \\0\end{pmatrix}}

	\newcommand{\Bound}{\partial \Omega}
	
	\date{}
	
\begin{document}
		\author{Ivan Pombo}
		\address{ Department of Mathematics\\
			University of Aveiro\\
			Campus Universit\'ario de Santagio\\
			3810-193 Aveiro, Portugal}
		\email{ivanpombo@ua.pt}
		\thanks{Department of Mathematics, Aveiro University, Aveiro 3810, Portugal.  This work was supported by Portuguese funds through CIDMA -
			Center for Research and Development in Mathematics and Applications and the Portuguese Foundation for Science and Technology 
			(``FCT--Funda\c{c}\~{a}o para a Ci\^{e}ncia e a Tecnologia''), within project UID/MAT/04106/2019}
		
		\title{CGO-Faddeev approach for Complex  Conductivities with Regular Jumps in two-dimensions}
		
		\begin{abstract}
			Researchers familiar with the state of the art are aware that the development of close-formed solutions for the EIT problem was not able to overpass the case of once-time differentiable conductivities beside the well known particular Astala-P\"aiv\"arinta result for zero frequency.  
			
			In this paper, we introduce some new techniques for the inverse conductivity problem combined with a transmission problem and achieve a reconstruction result based on an adaptation of the scattering data. The idea for these techniques, in particular the concept of admissible points is coming from E. Lakshtanov and B. Vainberg. Moreover, we are going to establish the necessary groundwork for working with admissible points which will be required in any further research in this direction.
			
		\end{abstract}
		
		\maketitle
		

		\section{Introduction}
		
		Consider $\mathcal O$ to be a bounded connected domain in $\mathbb R^2$ with a smooth boundary. The electrical impedance tomography (EIT) problem (e.g., \cite{borcea}) concerns the determination of the admittivity in the interior of $\mathcal O$, given simultaneous measurements of direct or alternating electric currents and voltages at the boundary $\partial \mathcal O$. 
		If the magnetic permeability is negligible, then the problem can be reduced to the inverse conductivity problem (ICP), which consists of reconstructing a function $\gamma(z), z \in \mathcal O$, via the known, dense in some adequate topology, set of data $(u|_{\partial \mathcal O},\frac{\partial u}{\partial\nu}|_{\partial \mathcal O})$, where
		\begin{equation}\label{set27A}
		\mbox{div}(\gamma \nabla u(z)) =0, ~  z\in \mathcal O.
		\end{equation}
		Here $\nu$ is the unit outward normal to $\partial \mathcal O$, $\gamma(z) = \sigma(z)+ i \omega\epsilon(z)$, where $\sigma$ is the electric conductivity and $\epsilon$ is the electric permittivity. If the frequency $\omega$ is sufficiently 
		small, then one can approximate $\gamma$ by a real-valued function.
		
		Previous approaches for EIT with isotropic admittivity, which were in active use for the last three decades,  can be divided in two groups: closed-form solution or sample methods where we refer to reviews \cite{borcea} and \cite{Pott} for further details as well as latest articles \cite{Aless1,Aless2,afr,BIY15,bukh,BKL,GL,GLSS,HT,knud,LL,lnv,lbcond,ltv,T}. These approaches do  not fully coincide, 
		for example the Linear Sampling Method (LSM) allows the reconstruction of the parameter's jump location, but it assumes the medium is known outside of the jump. An example of a weak point for the actual closed-form methods, i.e. Complex Geometric Optics (CGO) methods,  is that they do not allow for the presence of impenetrable obstacles anywhere inside the medium.

		Another problem which appears in the case of complex conductivities is the existence of exceptional points, i.e. non-trivial scattering solutions to the Lippmann-Schwinger equation  - roughly speaking, points where the solution for a given spectral parameter is not unique. Most methods for the inverse conductivity problem require the  condition that such exceptional points cannot occur (see, for example, \cite{nachman}). First ideas on how to handle the case of exceptional points appear in \cite{Novikov} and further in \cite{lnv}, \cite{lbcond}.

		For several years E. Lakshtanov and B. Vainberg had a parallel work with Armin Lechleitner on topics 
		like Interior Transmission Eigenvalues, inside-outside duality and factorization methods. In autumn of 2016 Armin wrote to them "...I'm actually not sure whether it pays off to develop these sampling methods further and further but I would be more interested in having methods for background media. We might try to continue our work in this direction, or towards Maxwell's equations, if this makes sense to you...". Although, the factorization methods are quite stable relatively to measurement errors, they fail if the outside medium is not known exactly but approximately. Armin Lechleitner obtain several results 
		in this direction, e.g. \cite{BKL}, \cite{GL}. In turn, E. Lakshtanov, R. Novikov and B. Vainberg also got some closed-form reconstruction/uniqueness results \cite{lnv},\cite{lbcond}. Furthermore, E. Lakshtanov, and B.Vainberg  got a feeling that LSM and CGO methods can be applied simultaneously to
		reconstruct the shape of the jump even if the potential is unknown. This lead to the new ideas being presented in the current paper, first among them the concept of admissible points. It is our believe that this concept will be an important step on how to proceed in the case of non-zero frequencies.
		
		The author would like to point out that the main ideas in this paper are from E. Lakshtanov and B. Vainberg who due to life circumstances were unable to pursue this line of research. The author is deeply indebted to them for allowing him to work out the details.
		
		As the methods for 2D and 3D are quite different even at the level of Faddeev Green function analysis, we focus our analysis on the 2D case only. Although, future plans are to extend the machinery we will present in order to obtain similar results in the 3D case.
		Moreover, in this paper we treat the isotropic case for complex conductivities with a jump. Recent results on the anisotropic case with real piecewise constant conductivities can be found in \cite{Aless3,Aless4}. One further extension of our approach could be to consider the anisotropic case with complex conductivities based on the previously mentioned works. 
		 
		We suppose that the conductivity function $\gamma$ is somehow smooth (to determine later) except in a closed contour $\Gamma \Subset \mathcal O$. Let $\gamma^+$ be the trace of $\gamma$ at the exterior part of the $\Gamma$ and $\gamma^-$ be the trace at the  interior part. By $\mathcal{D}$ we denote the interior part of $\Gamma$. 
		
		Under our assumption on $\gamma$ we look at solutions of the problem (\ref{set27A}) which are quite smooth in each domain, $u^- \in \mathcal{D}$ and $u^+ \in \mathcal O \backslash \overline{\mathcal D},$ and satisfy the following condition at $\Gamma$
		
		\begin{equation}\label{TransCond}
		\left \{
		\begin{array}{l}
		u^-(z)-u^+(z)=0, \quad \quad \quad \\
		\gamma^- \frac{\partial u^-}{\partial \nu}(z) - \gamma^+\frac{\partial u^+}{\partial \nu}(z) =0,  \quad \quad 
		\end{array}
		z \in \Gamma.
		\right .
		\end{equation}
		\vspace{0.2cm}
		
		The purpose of this approach is to establish a new method to overcome the limitation of Lipschitz conductivities in the current literature. In particular, we have in mind the handling of cases where separation of tissues is an important issue, like in detection of nodules through medical imaging.
		
		The reconstruction procedure of $\gamma$ starts by converting the conductivity similarly to \cite{bu} and \cite{fr}. Let $u$ be a solution of (\ref{set27A}) but only on the domain $\,\mathcal{O}\setminus\Gamma\,$ satisfying  the transmission condition above (\ref{TransCond}). 
		Below $z$ denotes a point in the complex plane and $\mathcal{O}$ is a domain in $\mathbb{C}$. Let $\partial=\frac{1}{2}\left(\frac{\partial}{\partial x}-i\frac{\partial}{\partial y}\right)$. Then the pair
		
		\begin{align}\label{0707B}
		\phi=(\phi_1, \phi_2)=\gamma^{1/2}(\partial u, \bar{\partial} u)^t=\gamma^{1/2}\left(
		\!\!\!\begin{array}{c}
		\partial u \\
		\bar{\partial} u\\
		\end{array} \!\!\!
		\right)
		\end{align}
		satisfies the Dirac equation
		\begin{equation}\label{firbc}
		\left ( \begin{array}{cc} \bar{\partial} & 0 \\ 0  & \partial \end{array} \right ) \phi(z) = q(z) \phi(z), \quad z=x+iy  \in \mathbb{C}\setminus\Gamma,
		\end{equation}
		with the potential $q$ defined also in $\mathbb{C}\setminus\Gamma$ by
		\begin{eqnarray}\label{char1bc}
		q(z)=\left ( \begin{array}{cc}0 &q_{12}(z) \\
		q_{21}(z) &  0\end{array} \right ), \quad q_{12}=-\frac{1}{2}\partial \log \gamma, \quad q_{21}=-\frac{1}{2}\bar{\partial }\log \gamma,
		\end{eqnarray}
		where we extend, as usual, $\gamma$ to the outside of  $\mathcal{O}$ by setting $\gamma=1$. On $\Gamma$, the pair $\phi$ satisfies a transmission condition for the Dirac equation which is derived from the previous one and we show it below.
		
		Thus, it is enough to solve the inverse Dirac scattering problem instead of the ICP. If it is solvable and $q$ can be found then the conductivity $\gamma$ is immediately obtained from (\ref{char1bc}), up to a constant. 
		In order to complete the reduction of the ICP to the inverse Dirac problem, one needs only to obtain the scattering data for the Dirac equation via the set of data $\left (u\!\left |_{\partial \mathcal O}\right .,\frac{\partial u}{\partial\nu}\!\left |_{\partial \mathcal O}\right. \right).$ 
				
		In fact, the scattering data for the Dirac equation can be obtained by simple integration of its Dirichlet data against the conjugate of an entire function $U$, which is related to the new set of complex geometric optic asymptotics, i.e.
		for a spectral parameter $\lambda$ and $w$ a certain type of point to be introduced, we have the scattering data to be define by
		\begin{align*}
			h(\lambda,w)= \int_{\partial \mathcal O} \overline{U(z,w,\lambda)} \,e^{-\overline{\lambda (z-w)^2}/4}\phi_2(z,w,\lambda)d\bar{z}.
		\end{align*}

		In this paper we give a reconstruction formula of the potential $q$ in the so-called admissible points (see Theorem~\ref{210918T}). We announce the result here in terms of a uniqueness theorem first since it does not require the introduction of the formal definition of the scattering data.
		We assume  that $\log \gamma$ is well defined in the whole complex plane, by assuming that the real part of the conductivity has a positive lower bounded.

		We have to remark that, in fact, we are going to present only a partial result, given that we cannot yet reconstruct, and show uniqueness of, the potential $q$ in the whole of $\mathbb{C}\setminus\Gamma$. Our proof is based on a new concept, that is based on a specific set of points, which we now define:

		\begin{defn} \label{PontosAdmissiveis}
			
			We say that a point $w\in\mathcal{O}$ is an {\it admissible point} if there is a number $\lambda_{\mathcal O} \in \mathbb{C}$ such that
			\begin{align*}
			A&:=\sup_{z \in \overline{\mathcal O}}\, \textnormal{Re}[ \lambda_{\mathcal O}(z-w)^2]< 1/2\\
			B&:= \sup_{z \in \overline{\mathcal D}} \, \textnormal{Re}[ \lambda_{\mathcal O}(z-w)^2] <  -1/2.
			\end{align*}
			
			Moreover, if $w$ is an admissible point and the constants $A$ and $B$ fulfills
			
			$A=1/2-\epsilon_1, \, B=-1/2-\epsilon_2,$ with $\epsilon_2-\epsilon_1>0$, we further say that $w$ is a proper admissible point.
			
		\end{defn}

		The main theorem of this chapter will be obtained using this novel idea. Even though the proof follows by reconstruction, we give here the uniqueness theorem without introducing the scattering data first, given that this will be related with new CGO incident waves.
		\begin{theorem}
			
			Let $\Omega$ be a bounded Lipschitz domain in the plane, and let $\gamma\in W^{2,\infty}(\mathcal{D})\cap W^{2,\infty}(\mathcal{O} \setminus\overline{\mathcal{D}})$ such that $\textnormal{Re}(\gamma)\geq c>0$. If $\sqrt{\frac{\gamma^-}{\gamma^+}}-1$ is small enough on $L^{\infty}(\Gamma)$, we have that the Dirichlet-to-Neumann map $\Lambda_\gamma$ determines the conductivity $\gamma$ uniquely in any proper admissible point.
		\end{theorem}
		
		Hereby, we want to point out that the Dirichlet-to-Neumann map determines the scattering data uniquely can be proven similarly to~\cite{lbcond} (see Section \ref{ScDN}).
		\\

		Now, Theorem~\ref{210918T} will even provide a reconstruction formula for the potential $q$ in so-called proper admissible points. This is an improvement of previous existent methods insofar as a convenient enlargement of the set of CGO incident waves allows to highlight the desirable areas around such points. Thus, this article provides a 2D reconstruction result for complex conductivities which are discontinuous on a contour which, although being apparently a rather weak result, cannot possibly be obtained by any previous technique, at least that we know of, and represents a first step in this direction. In fact the main goal of the article is to show the viability of the presented approach. In this manner, all our efforts are to present the main tools for this approach, leaving other questions like, stability of determination as in \cite{Aless}, many contours and geometry of admissible points, to future work.
	
		We also want to point out that our definition of admissible point is not sharp, i.e. it can be made sharper by considering higher regularity of the conductivity outside the curves of discontinuities $\Gamma$.

		Several technical problems need to be solved and presented now in order to facilitate the subsequent study. These include: the right choice of the functional space, a set of admissible points (essential to the reconstruction), and the enrichment of the set of CGO incident waves (i.e. we use solutions like $|\lambda|^{f(z)}$ which highlight desirable areas). The latter solutions are unlimited even after the CGO-Faddeev normalization and we are required to obtain two-dimensional Laplace Transform analogues of the Hausdorff-Young inequality to derive our reconstruction formula. \\
		
		The paper is organized as follows: In Section 2 we recall necessary facts on the transmission condition and the construction of the Lippmann-Schwinger equation for CGO-Faddeev solutions in our case. In Section 3. we introduce the necessary function spaces as well as related lemmas. We present the novel concept of \textit{admissible points} (see Definition \ref{PontosAdmissiveis}) based on a convenient enrichment of the set of CGO incident waves and we study the scattering data and reconstruction of the potential in these type of points. We finalize this section with two subsections containing some more necessary results and the proof of our main theorem. For the sake of readability we placed some additional results together with its proofs in an appendix.

\section{Main construction}
    	
    	\subsection{Transmission condition}
    	We denote by $n(z)=\left(n_x(z),n_y(z)\right)$ the unit outer normal vector in $\Gamma$ and on the complex plane by $\nu(z)=n_x(z)+in_y(z)$.
    	During the paper we consider two orientations for the contour $\Gamma$: positively oriented $\Gamma^+$ (curve interior $\mathcal{D}$ is to the left) and  negatively oriented $\Gamma^-$ (curve interior is to the right).
    	\begin{lemma}
    		The transmission condition (\ref{TransCond}) implies the following condition to the Dirac equation on $\Gamma$
    		\begin{equation}\label{7JunB}
    		\left(
    		\begin{array}{c}
    		\phi_1^+ - \phi_1^- \\
    		\phi_2^+ - \phi_2^- \\
    		\end{array}
    		\right)
    		=
    		\frac{1}{2}\left(
    		\begin{array}{cc}
    		\alpha+\frac{1}{\alpha}-2 & (\alpha-\frac{1}{\alpha})\bar{\nu}^2 \\
    		(\alpha-\frac{1}{\alpha})\nu^2 & \alpha+\frac{1}{\alpha}-2 \\
    		\end{array}
    		\right)
    		\left(
    		\begin{array}{c}
    		\phi_1^- \\
    		\phi_2^- \\
    		\end{array}
    		\right)
    		\end{equation}
    		where $\alpha=\sqrt{\frac{\gamma^-}{\gamma^+}}$.
    	\end{lemma}
    	\begin{proof}
    	Let $l(z)=(-n_y(z),n_x((z))$ be a unit tangential vector to $\Gamma$.
    	From the first equation of (\ref{TransCond}) follows for the tangential derivative that $\frac{\partial}{\partial l} (u^+(z)-u^-(z))=0$ and, therefore, 
    	$$
    	\sqrt{\gamma^+} u_l^+ - \sqrt{\gamma^-} u_l^- =u_l^- \sqrt{\gamma^-} (\frac{1}{\alpha}-1),
    	$$
    	where $u_l=\frac{\partial u}{\partial l}$. Moreover, during this proof and to simplify the computations we denote the normal derivative as  $u_n=\frac{\partial u}{\partial \nu}$. From the second equation of (\ref{TransCond}) we get $u_n^+ =\frac{\gamma^-}{\gamma^+}  u_n^- $, where $u_n^{\pm}$ denotes the normal derivative of $u^{\pm}$, so that 
    	$$
    	\sqrt{\gamma^+} u_n^+ - \sqrt{\gamma^-} u_n^-=\sqrt{\gamma^-} u_n^- (\alpha-1).
    	$$
    	Note that we have now
    	\begin{eqnarray} 
    	{\partial}u = \frac{1}{2} (\bar{\nu} u_n - i\bar{\nu} u_l), \label{8} \\
    	\bar{\partial}u = \frac{1}{2} (\nu u_n + i\nu u_l), \label{9}
    	\end{eqnarray}
    	\begin{align*}
    	\phi_1^+-\phi_1^-= \sqrt{\gamma^+} {\partial}u^+ - \sqrt{\gamma^-} {\partial}u^-=
    	\left(\frac{1}{\alpha}-1\right) u_l^- \sqrt{\gamma^-} \frac{1}{2}(-i\bar{\nu}) + ({\alpha}-1) u_n^- \sqrt{\gamma^-} \frac{1}{2}\bar{\nu},   \\
    	\phi_2^+-\phi_2^-= \sqrt{\gamma^+} \bar{\partial}u^+ - \sqrt{\gamma^-} \bar{\partial}u^-=
    	\left(\frac{1}{\alpha}-1\right) u_l^- \sqrt{\gamma^-} \frac{1}{2}(i\nu) + \left({\alpha}-1\right) u_n^- \sqrt{\gamma^-} \frac{1}{2}\nu.
    	\end{align*}
    	
    	These relations take the matricial form 
    	
    	$$
    	\left(
    	\begin{array}{cc}
    	\phi_1^+-\phi_1^-\\
    	\phi_2^+-\phi_2^-\\
    	\end{array}
    	\right) = \frac{1}{2}\left(
    	\begin{array}{cc}
    	(\alpha-1)\bar{\nu} & (\frac{1}{\alpha}-1) (-i\bar{\nu}) \\
    	(\alpha-1){\nu} & (\frac{1}{\alpha}-1)(i{\nu}) \\
    	\end{array}
    	\right)\left(
    	\begin{array}{cc}
    	u_n^- \sqrt{\gamma^-}  \\
    	u_l^- \sqrt{\gamma^-}  \\
    	\end{array}
    	\right).
    	$$
    	Using (\ref{8}) and (\ref{9}), together with the definition of $\phi,$ we obtain the relation
    	$$
    	\left(
    	\begin{array}{cc}
    	u_n^- \sqrt{\gamma^-}  \\
    	u_l^- \sqrt{\gamma^-}  \\
    	\end{array}
    	\right)=
    	\left(
    	\begin{array}{cc}
    	{\nu} & \bar{\nu} \\
    	i\nu & -i\bar{\nu} \\
    	\end{array}
    	\right)\left(
    	\begin{array}{cc}
    	\phi_1^-  \\
    	\phi_2^-\\
    	\end{array}
    	\right).
    	$$
    	These two previous displayed equations allows us to complete the proof of the lemma.
    \end{proof}
%
%
%
    	
    	\subsection{The Lippmann-Schwinger equation for CGO-Faddeev solutions}
    	Consider the vector $\phi$ which satisfies (\ref{firbc}) and the following asymptotic
    	\begin{align}\label{8JunA}
    	\begin{array}{l}
    	\phi_1(z,w,\lambda)=e^{\lambda (z-w)^2/4}U(z,w,\lambda)+e^{\lambda (z-w)^2/4}o(1), \quad\\ \phi_2(z,w,\lambda)=e^{\overline{\lambda (z-w)^2}/4}o(1), 
    	\end{array} \quad z \rightarrow \infty.
    	\end{align}
    	where $U(z,w,\lambda)$ is an entire function with respect to the parameter $z$.
    	
    	We denote  
    	\begin{equation}\label{200918A}
    	\mu_1(z,w,\lambda)=\phi_1(z,w,\lambda) e^{-\lambda (z-w)^2/4}, \quad \mu_2(z,w,\lambda)=\phi_2(z,w,\lambda)e^{-\overline{\lambda (z-w)^2}/4}.
    	\end{equation}
    	
    	Further, we introduce some matrix functions that will establish a integral equation for $\mu$.
    	
    	Due to (\ref{firbc}), this functions fulfill the following equation on $\mathbb{C}\setminus\Gamma$:
    	\begin{equation}\label{muirbc}
    	\left ( \begin{array}{cc} \bar{\partial}_z & 0 \\ 0  & \partial_z \end{array} \right ) \mu =\left( \begin{array}{cc}
    	0 & q_{12}(z)e^{-i\,\text{Im}[\lambda(z-w)^2/2]} \\
    	q_{21}(z)e^{i\,\text{Im}[\lambda(z-w)^2/2]}
    	\end{array}\right) \mu =: \tilde{q}\mu.
    	\end{equation}
    	
    	On the contour $\Gamma$ they fulfill a transmission condition similar to (\ref{7JunB}), with the right-hand side being substituted by:
    	$$
    	\widetilde{A}_\lambda\mu=\frac{1}{2}\left(
    	\begin{array}{cc}
    	\alpha+\frac{1}{\alpha}-2 & (\alpha-\frac{1}{\alpha})\bar{\nu}^2 e^{-i\,\textnormal{Im}[\lambda (z-w)^2/2]}\\
    	(\alpha-\frac{1}{\alpha})\nu^2e^{i\,\textnormal{Im}[\lambda (z-w)^2/2]} & \alpha+\frac{1}{\alpha}-2 \\
    	\end{array}
    	\right)\,
    	\left(\begin{array}{cc}
    	\mu_1^- \\
    	\mu_2^-
    	\end{array}\right),
    	$$
    	where $\mu_1^-$ and  $\mu_2^-$ are the traces values of $\mu$ taken from the interior of $\Gamma$.
    	
    	Through this we obtain an integral equation for $\mu$:
    	\begin{proposition}
    		Let $\mu$ be a solution of (\ref{muirbc}) given as above through a function $\phi$ which fulfills (\ref{firbc}) and the asymptotics (\ref{8JunA}). Then $\mu$ is a solution of the following integral equation:
    		\begin{equation}\label{0807A}
    		(I+P\widetilde{A}_\lambda -D\widetilde{Q}_\lambda )\mu=   \left(
    		\begin{array}{c}
    		U \\
    		0 \\
    		\end{array}
    		\right),
    		\end{equation}
    		where $D=\left(
    		\begin{array}{cc}
    		\bar{\partial}^{-1} & 0 \\
    		0 & \partial^{-1} \\
    		\end{array}
    		\right)$ with $ \bar{\partial}^{-1}f(z)=\frac{1}{2\pi i}\int_{\mathbb{C}} \frac{f(\varsigma)}{\varsigma-z}\,d\varsigma\wedge d\bar{\varsigma} 
    		$
    		and the $\partial^{-1}$ is given through the complex conjugate of the kernel $(\varsigma-z)^{-1}$. The matrix $\widetilde{Q}_\lambda$ has the following form 
    		
    		$$
    		\widetilde{Q}_\lambda=\left(
    		\begin{array}{cc}
    		0 & Q_{12} e^{-i\,\textnormal{Im}[\lambda (z-w)^2/2]}\\
    		Q_{21} e^{i\,\textnormal{Im}[\lambda (z-w)^2/2]} & 0 \\
    		\end{array}
    		\right),
    		$$ where $Q_{12},\,Q_{21}$ are $L^{\infty}$ extensions of $q_{12},\,q_{21}$ to $\Gamma$.
    		Moreover, $P$ is a projector 
    		\begin{equation}\label{14a}
    		P=\left(
    		\begin{array}{cc}
    		P_+ & 0 \\
    		0& P_- \\
    		\end{array}
    		\right),
    		\end{equation} where $P_+,\,P_-$ are the Cauchy projector and its complex adjoint, respectively:
    		$$
    		P_+f(w)=\frac{1}{2\pi i} \int_{\Gamma^+} \frac{ f(z)}{z-w}\,dz, \quad P_-f(w)=\frac{1}{2\pi i} \int_{\Gamma^+} \frac{ f(z)}{\bar{z}-\bar{w}}\,d\bar{z},  \quad w \in \mathbb C.
    		$$
    		Hereby, $f$ is a function defined on the contour $\Gamma$.
    	\end{proposition}
    	   
    	\begin{proof} We use the same approach as in \cite{lbcond}. The following Cauchy-Green formulas hold for each $f\in C^1(\overline{\Omega})$ and an arbitrary bounded domain $\Omega$ with smooth boundary:
    		\begin{eqnarray}\label{2DecB}
    		f(z)=  \frac{1}{2\pi i} \int_{ \Omega} \frac{\partial f (\varsigma)}{\partial \bar{\varsigma}} \frac{1}{\varsigma-z} \,d\varsigma\wedge d\bar{\varsigma} + \frac{1}{2\pi i} \int_{\Bound} \frac{f(\varsigma)}{\varsigma - z}d\varsigma, \quad z \in  \Omega, \\ \label{2DecC}
    		0=  \frac{1}{2\pi i} \int_{\Omega} \frac{\partial f (\varsigma)}{\partial \bar{\varsigma}} \frac{1}{\varsigma-z}\,d\varsigma\wedge d\bar{\varsigma} + \frac{1}{2\pi i} \int_{\Bound} \frac{f(\varsigma)}{\varsigma - z}d\varsigma,\quad z \not \in \overline{\Omega}.
    		\end{eqnarray}

    		Denote by $D_R$ a disk of radius $R$ and centered at z, and take  $D_R^-=D_R\setminus \overline{\mathcal{D}} $. We recall, $D$ is the interior part of $\Gamma$. Assume that $z\in \mathcal{D}$, $f=\mu_1$ in both formulas, and $\Omega=\mathcal{D}$ in (\ref{2DecB}) and $\Omega=D_R^-$ in (\ref{2DecC}). We add the left- and right-hand sides in formulas (\ref{2DecB}) and (\ref{2DecC}). Taking the transmission condition for $\mu$ into account, we obtain for fixed $w$ that
    		\begin{equation}\label{11NovA}
    		\mu_1(z,\lambda) = \frac{1}{2\pi i} \int_{D_R\setminus\Gamma}  (\tilde{q}\mu)_1(\varsigma,\lambda) \frac{1}{\varsigma-z}\,d\varsigma\wedge d\bar{\varsigma} + \frac{1}{2\pi i} \int_{\Gamma^-} \frac{ [\mu_1](\varsigma)}{\varsigma - z}d\varsigma+ \frac{1}{2\pi i} \int_{\partial D_R} \frac{\mu_1(\varsigma)}{\varsigma - z}d\varsigma,
    		\end{equation}
    		where $[\mu_1]=\mu_1^- -\mu_1^+$.
    		
    		Noticing that $\mu_1$ converges to  $U$ at infinity and since $U$ is entire, then taking the limit $R\to \infty$, it follows that the last term is $U(z)$. In this way, by taking the limit, and reordering we obtain:
    		\begin{equation}
    			\mu_1(z,\lambda) - \frac{1}{2\pi i} \int_{\mathbb{C}} (\widetilde{Q}_\lambda\mu)_1(\varsigma,\lambda)\frac{1}{\varsigma-z} \,d\varsigma\wedge d\overline{\varsigma} +\frac{1}{2\pi i}\int_{\Gamma^+} \frac{(\widetilde{A}_{\lambda}\mu)_1(\varsigma,\lambda)}{\varsigma-z}\,d\varsigma = U(z).
    		\end{equation}
    		
    		This equation together with similar computations for $z \in D_R^{-}$ and showing the case for $\mu_2$ (similarly by taking the adjoint Cauchy-Green formulas) we obtain the desired integral equation.\end{proof}
    	
    	\section{Technical details}
    	
    	\subsection{The choice of the function space}
    	Let  $1<p<\infty$, $R>0$ and $f = \left(\begin{array}{cc}
    	f_1 \\ f_2
    	\end{array}\right)$ be a vector function.
    	
    	To define our spaces we keep in mind the notation introduced in \cite{ltv}. Denote by $L^{\infty}_z(B)$ the space of bounded functions of $z\in \mathbb{C}$ with values in a Banach Space $B$. Thus, picking $B=L^p_{\lambda}(|\lambda|>R)$ we introduce the first space
    	$$\mathcal H_1^p:=\left\{f:  f_1,\,f_2\;\, \textnormal{continuous functions in}\, L^{\infty}_z\left(L^p_{\lambda}(|\lambda|>R)\right) \cap L^{\infty}_z\left(L^{\infty}_{\lambda}(|\lambda|>R)\right)\right\}.$$

    	 To simplify the notation ahead, we introduce the following function space:
    	$$S=\left\{g:\Gamma\times\{\lambda\in\mathbb{C}:|\lambda|>R\}\rightarrow \mathbb{C}^2 \;\; \text{s.t.}\; \sum_{i\in\{1,2\}} \int_{|\lambda|>R}
    	\int_{\Gamma} |g_i(z,\lambda)|^p\,d|z| d\sigma_\lambda < \infty\right\},$$
    	where $d\sigma_{\lambda}$ is the Lebesgue measure in $\mathbb{R}^2$ (similarly we define $d\sigma_{z}$).
    	
    Following the idea of Hardy spaces and to obtain desirable properties at the contour $\Gamma$ we define the second space through the projector $P$ in (\ref{14a}) by: 
    	\begin{align*}
    		\mathcal H_2^p:= \left\{F\in \mathcal{R}(P): \sum_{i\in\{1,2\}} \int_{|\lambda|>R}\int_{\Gamma} \left|F^-_i(z,\lambda)\right|^p\,d|z|d\sigma_{\lambda}<\infty\right\}
    	\end{align*}
    	where by $\mathcal{R}(P)$ we mean the range of the matrix projector $P$ with domain $S$. Hence we have that for $F \in \mathcal{H}^p_2$ there exists a function $f\in S$ such that $F=Pf$ and in $\Gamma$ it fulfills $F^-=f$. Moreover, this allows to consider this space with the norm
    	$$\|F\|_{\mathcal{H}^p_2}^p:=\sum_{i\in\{1,2\}} \int_{|\lambda|>R} \int_{\Gamma} |F^-_i(z,\lambda)|^p\,d|z| d\sigma_\lambda= \sum_{i\in\{1,2\}} \int_{|\lambda|>R} \int_{\Gamma} |f_i(z,\lambda)|^p\,d|z| d\sigma_\lambda.$$
    	
    	Finally, the space we are going to work with is given as $\mathcal H^p = \mathcal H^p_1 + \mathcal H^p_2$  endowed with the norm
    	\begin{equation}\label{2107B}
    	\|t\|_{\mathcal H^p}=\inf_{\substack{u+v=t \\u \in \mathcal H^p_1, v \in \mathcal H^p_2}} \max (\|u\|_{\mathcal H^p_1}, \|v\|_{\mathcal H^p_2 }).
    	\end{equation}
    	Let us remind that the operations of intersection and union of two Banach spaces are correctly defined if all terms can be continuously embedded into a common locally convex space. In our situation this common locally convex space will be a space endowed with the semi-norms 
    	$$
    	 \int_{|\lambda|>R} \int_{ \mathcal O} \frac{1}{|\lambda|^2}|f(z,\lambda)|\,d\sigma_zd\sigma_\lambda.
    	$$
    	
    	If $f \in \mathcal H_1^p$ the embedding is evident. For $f \in \mathcal H_2^p$ we have 
    	$$
    	\|Pf\|_{L^p(\mathcal O)} \leq C \|f\|_{L^p(\Gamma)},
    	$$
    	so that $\left[\|Pf\|_{L^p(\mathcal O)} \right]^p \leq \left[\|f\|_{L^p(\Gamma)}\right]^p$
    	and
    	$$
    	\int \left ( \int_{\mathcal O} |Pf(z)|^p \,d\sigma_z \right ) \,d\sigma_\lambda \leq  \int [\|f\|_{L^p(\Gamma)}]^p\, d\sigma_\lambda = \int_{|\lambda|>R} \int_{\Gamma} |f(z,\lambda)|^p \,d|z|d\sigma_\lambda.
    	$$
    	The boundedness of each semi-norm follows from the continuity of the embedding of $L^p(\mathcal O)$ into $L^1(\mathcal O)$. 
    	
    	\begin{lemma}\label{lemma2107D}
    		The operators $\widehat{P}_\pm: f \rightarrow (Pf)|_{\Gamma^\pm}$ are bounded in the space with norm
    		$$\left[\int_{|\lambda|>R} \int_{\Gamma} |f(z,\lambda)|^p\,dz d\sigma_\lambda\right]^{1/p}.$$
    	\end{lemma}
    	
    	\begin{proof}
    		 During the proof the sign $\pm$ in the projectors will be omitted. From the continuity of Cauchy projectors in $L^p(\Gamma)$ follows 
    		$$
    		\|\widehat{P} f \|_{L^p(\Gamma)} \leq C \|f \|_{L^p(\Gamma)}.
    		$$
    		and therefore
    		$$
    		\left (\|\widehat{P} f \|_{L^p(\Gamma)} \right )^p\leq C^p \left ( \|f \|_{L^p(\Gamma)} \right )^p.
    		$$
    		Finally
    		\begin{small}
    			$$\|P\widehat{P}f\|_{\mathcal H^p_2}= \int_{|\lambda|>R} \left (\|\widehat{P} f \|_{L^p(\Gamma)} \right )^p \,d\sigma_\lambda \leq C^p \int_{|\lambda|>R} \left( \|f \|_{L^p(\Gamma)} \right)^p \,d\sigma_\lambda = C^p \int_{|\lambda|>R} \int_{\Gamma} |f(z,\lambda)|^p \,d|z|d\sigma_\lambda .$$
    		\end{small}
    	\end{proof}
    	\begin{lemma}\label{lemma2107C}
    		Let $u \in \mathcal H^p_1$. Then $P(u|_{\Gamma}) \in \mathcal H^p_2$.
    	\end{lemma}
    	\begin{proof}
    		From the definition of $\mathcal H^p_1$, combined with the fact that $u$ is a continuous function, we get
    		$$
    		\|u\|_{L^p_\lambda} \in L^\infty_z(\Gamma).
    		$$
    		Since $\Gamma$ is a bounded set, the $L^p$ norm does not exceed (up to a constant) the $L^\infty$ norm and, therefore
    		$$
    		\|\|u\|_{L^p_\lambda}\|_{L^p_z(\Gamma)} \leq C \|u\|_{\mathcal H^p_1}.
    		$$
    		Now we just note the left-hand side of the above inequality is the norm  $\mathcal H^p_2$ norm. 
    	\end{proof}
    	
    	\subsection{Analysis of the Lippmann-Schwinger equation}
    	Multiplying equation (\ref{0807A}) by $I+D\widetilde{Q}_\lambda $  we get
    	\begin{equation}\label{0907H}
    	(I+M)\mu= (I+D\widetilde{Q}_\lambda) \left(
    	\begin{array}{c}
    	U \\
    	0 \\
    	\end{array}
    	\right)
    	\end{equation}
    	where
    	\begin{equation}\label{0907I}
    	M=P\widetilde{A}_\lambda +D\widetilde{Q}_\lambda P\widetilde{A}_\lambda -D\widetilde{Q}_\lambda D\widetilde{Q}_\lambda.
    	\end{equation}
    	
    	\begin{lemma}\label{lemma2107A}
    		Let $\,\textnormal{Re}(\gamma)\geq c >0$ and $\widetilde{A}_{\lambda}$ in $L^{\infty}(\Gamma)$. Then the operators $D\widetilde{Q}_\lambda P\widetilde{A}_\lambda ,\,D\widetilde{Q}_\lambda D\widetilde{Q}_\lambda$ are bounded in  $\mathcal H^p,\, p>1$. 
    		
    		Moreover, if $R>0$ is large enough they are contractions and if $\alpha-1$ is small enough in $L^{\infty}(\Gamma)$ then $P\widetilde{A}_{\lambda}$ is a contraction in $\mathcal{H}^p, \,p>1$.
    	\end{lemma}
    	\begin{proof} In order to estimate $\|(D\widetilde{Q}_\lambda P\widetilde{A}_\lambda)t\|_{\mathcal H^p}$ and $\|(D\widetilde{Q}_\lambda D\widetilde{Q}_\lambda )t\|_{\mathcal H^p}$ (recall Definition \ref{2107B}) we consider the representation $t=u+v$ where the infimum is (almost) achieved. It is easy to see that the desired estimate follows from the fact that these operators are a contraction in each of the spaces, $\mathcal H_1^p$ and $\mathcal H_2^p$. This fact can be shown as follows. 
    	
    	In Lemma 2.1 of \cite{ltv} it was proved that the operator $D\widetilde{Q}_\lambda D\widetilde{Q}_\lambda$ is bounded in $\mathcal H^p_1$. The proof that it is also a contraction in $\mathcal{H}^p_2$ and the statement for $D\widetilde{Q}_\lambda P\widetilde{A}_\lambda$ follows in a similar manner. 
    	
    	Hereby, we show the case for $D\widetilde{Q}_\lambda P \widetilde{A}_{\lambda}$. By definition we have:
    	$$
    	D\widetilde{Q}_\lambda P\widetilde{A}_{\lambda} u(z)=\left\{\begin{array}{ccc}
    	\int_{\Gamma} [\widetilde{A}_{\lambda}u]_2(z_2)G_1(z,z_2,\lambda,w)\,dz_2 \\
    	\int_{\Gamma} [\widetilde{A}_{\lambda}u]_1(z_2)G_2(z,z_2,\lambda,w)\,dz_2
    	\end{array}\right.$$
    	where 
    	\begin{equation}\label{AAA}
    	G(z,z_2,\lambda,w)=\left(\begin{array}{cc}
    	G_1\\
    	G_2
    	\end{array}\right) =\left\{\begin{array}{cc}
    	(2\pi i)^{-2}\int_{\mathcal O}
    	\frac{e^{-i\,\textnormal{Im}[\lambda(z_1-w)^2]/2}}{{z_1}-{z}}  \frac{{Q}_{12}(z_1)}{\bar{z}_2-\bar{z}_1} \, d{\sigma_{z_1}}\\
    	(2\pi i)^{-2}\int_{\mathcal O}
    	\frac{e^{i\,\textnormal{Im}[\lambda(z_1-w)^2]/2}}{{\bar{z}_1}-\bar{z}}  \frac{{Q}_{21}(z_1)}{{z}_2-{z}_1} \, d{\sigma_{z_1}}
    	\end{array}\right. .
    	\end{equation}
    	 
    	By following a similar estimation on the proof of Lemma 2.1 of \cite{ltv} we obtain by the stationary phase approximation:
    	$$
    	\sup_{\substack{z\\|\lambda|>R}} \|G_i(z,\cdot,\lambda,w)\|_{L^q_{z_2}(\Gamma)} \leq \frac{1}{R}, ~  \quad 1/p+1/q=1 ~ \text{and} ~i=1,2.
    	$$
    	Thus
    	$$
    	|D\widetilde{Q}_\lambda P \widetilde{A}_{\lambda}u|(z)\leq \|G(z,\cdot,\lambda,w)\|_{L^q_{z_2}(\Gamma)}  \|\widetilde{A}_{\lambda}u\|_{L^p_{z_2}(\Gamma)}.
    	$$
    	Then, we have for 
    	$$
    	\|D\widetilde{Q}_\lambda P\widetilde{A}_{\lambda} u(z)\|_{L^p_\lambda} \leq  \|G(z,\cdot,\lambda,w)\|_{L^q_{z_2}(\Gamma)} \|\widetilde{A}_{\lambda}\|_{L^\infty(\Gamma)} \|u\|_{\mathcal H^p_2}
    	$$
    	where we used the fact that $u \in \mathcal H^p_2$ is the same as $ \|u\|_{L^p_{z_2}(\Gamma)} \in L^p_\lambda$. The final estimate follows from the definitions of both spaces and the above uniform bound on $G_i$.
    	
    	If we take $R>0$ large enough then it follows that $D\widetilde{Q}_{\lambda}D\widetilde{Q}_{\lambda}$ and $D\widetilde{Q}_{\lambda}P\widetilde{A}_{\lambda}$ are contractions in $\mathcal{H}^p$ as long as $\|\widetilde{A}\|_{L^{\infty}(\Gamma)}$ is finite.
    	
    	By the definition of $\mathcal{H}^p$ the boundedness of $P\widetilde{A}_{\lambda}$ follows from the usual $L^p$ boundedness. Since this operator will not have the same dependence on $\lambda$ as the others we need the jump to be close enough to $1$ so that the supremum norm in $z$ of $\widetilde{A}_{\lambda}$ on $\Gamma$ is small enough and possibilitates the norm of the whole operator to be less than $1$.
    	
    	A rough estimate for this norm is given in terms of the jump by:
    	 $$\|\widetilde{A}_{\lambda}\|_{L^{\infty}(\Gamma)}\leq 2\left|\alpha-1\right|\left(1+\left|\frac{1}{\alpha}\right|\right)\leq 4\epsilon,$$
    	where $\epsilon>0$ is an upper bound for
    	$\left|\alpha-1\right|$.
    	Hence for $P\widetilde{A}_{\lambda}$ to be a contraction on $\mathcal{H}^p$ we need that
    	$$|\alpha-1| \leq \frac{1}{4\|P\|_{\mathcal{H}^p}}.$$
    	\end{proof}
    	
    	\subsection{Enrichment of the set of CGO incident waves}
    	Let $w \in \mathcal O$ be a fixed point. For the asymptotics (\ref{8JunA}) we can take any entire function. In our approach, we take this entire functions to be:
    	\begin{equation}\label{0907A}
    	U(z,w,\lambda)=e^{\ln |\lambda| \lambda_{\mathcal O}(z-w)^2},
    	\end{equation}
    	where $z\in \mathbb{C}$ and $\lambda_{\mathcal O}$ is a parameter.
    	\\
    	These functions lead us to the concept of admissible points. 
    	\\ 
    	
    	We recall here their definition:     	
    		We say that a point $w \in \mathcal{O}$ is an {\it admissible point}, if there is a number $\lambda_{\mathcal O}\in\mathbb{C}$ such that
    		\begin{align*}\label{0907B}
    		A&:=\sup_{z \in \overline{\mathcal O}} \textnormal{Re}[ \lambda_{\mathcal O}(z-w)^2]< 1/2,\\
    		B&:= \sup_{z \in \overline{\mathcal D}} \textnormal{Re}[ \lambda_{\mathcal O}(z-w)^2] <  -1/2.
    		\end{align*}
    		Moreover, if $w$ is an admissible point and $A$ and $B$ fulfills
    		$A=1/2-\epsilon_1, \, B=-1/2-\epsilon_2,$ with $\epsilon_2-\epsilon_1>0$, we further say that $w$ is a proper admissible point.

    	{\bf Note:} The set of admissible points is not empty. In order to see this we consider a boundary point $w_0 \in \partial \mathcal O$ which belongs also to the convex hull of $\mathcal O$. It is easy to see that all interior points $w \in \mathcal O$ near the $w_0$ would be admissible. 
    	
    	We will not try to give a general geometric description of admissible points. Instead, we are only aiming to show the viability of the concept.   
    	
    	Denote
    	\begin{equation}\label{1107A}
    	f=\mu -  (I+D\widetilde{Q}_\lambda)\left(
    	\begin{array}{c}
    	U \\
    	0 \\
    	\end{array}
    	\right),
    	\end{equation}
    	where $\mu$ is defined in (\ref{200918A}).
    	
    	The vector $f$ satisfies the equation
    	\begin{equation}\label{1007P}
    	(I+M)f= -M(I+D\widetilde{Q}_\lambda) \left(
    	\begin{array}{c}
    	U \\
    	0 \\
    	\end{array}
    	\right).
    	\end{equation}
    	We know already that for $R>0$ large enough the operator in the left-hand side of this equation is a contraction in $\mathcal H^p,\, p>1$ and below we show that in fact we have for the right-hand side:
    	\begin{equation}\label{1007Q}
    	\frac{1}{|\lambda|^A} M(I+D\widetilde{Q}_\lambda) \left(
    	\begin{array}{c}
    	U \\
    	0 \\
    	\end{array}
    	\right) \in \mathcal H^p, \quad p>2.
    	\end{equation}
    	Therefore, we get the following statement
    	\begin{lemma}\label{lemma1107A}
    		For any $p>2$ and $R$ large enough such that $U$ is given in terms of a proper admissible point $w$ we have
    		\begin{equation}\label{0907F}
    		\frac{1}{|\lambda|^A} \left [ \mu -  (I+D\widetilde{Q}_\lambda)\left(
    		\begin{array}{c}
    		U \\
    		0 \\
    		\end{array}
    		\right) \right ] \in \mathcal H^p.
    		\end{equation}
    	\end{lemma}
    	\begin{proof}
    		We start by showing that 
    		\begin{equation}\label{Hp}\frac{1}{|\lambda|^A}\U \in \mathcal{H}_2^p
    		\end{equation} and $$\frac{1}{|\lambda|^A}MD\tilde{Q}_{\lambda}\U \in \mathcal{H}_1^p.$$ Since $M$ is a contraction for $R > 0$ big enough we are going to obtain $(\ref{1007Q})$ and the result will immediately follows for $p > 2$.
    		
    		To show (\ref{Hp}) we refer to the following simple estimate
    		\begin{align*}
    		\Bigg[\int_{|\lambda| >  R}\int_{\Gamma}\Bigg|\frac{1}{|\lambda|^A}e^{\ln |\lambda|\lambda_s(z-w)^2}\Bigg|\,d|z|\,d\sigma_{\lambda}\Bigg]^{1/p} 	
    		&\leq \Bigg[\int_{|\lambda| >  R}\int_{\Gamma}\Bigg|\frac{1}{|\lambda|^A} |\lambda|^B\Bigg|^p\,d|z|\,d\sigma_{\lambda}\Bigg]^{1/p}
    		\\
    		&
    		= \Bigg[\int_{|\lambda| >  R}\int_{\Gamma}\Bigg|\frac{1}{|\lambda|^{1+(\epsilon_2-\epsilon_1)}} \Bigg|^p\,d|z|\,d\sigma_{\lambda}\Bigg]^{1/p}  < \infty
    		\end{align*}
    		
    		For the second statement we need to dismantle $M$ into its various parts and show that the statement holds for each one of them. The trick is always the same, so we will only show one of the computations, namely the one corresponding to the term $\frac{1}{|\lambda|^A}(D\tilde{Q}_{\lambda})^3\U$. By Lemma \ref{2310F} we get
    		
    		\begin{align*}
    		&\sup_{z \in \mathcal O} \Bigg[\int_{|\lambda|>R}\Bigg|\frac{1}{|\lambda|^A}\int_{ \mathcal{O}}\frac{e^{i\,\textnormal{Im}(\lambda(z_1-w)^2)}}{\overline{z_1-z}}Q_{21}(z_1)\int_{ \mathcal{O}}\frac{e^{-i\,\textnormal{Im}(\lambda(z_2-w)^2)}}{z_2-z_1}Q_{12}(z_2) \cdot
    		\\
    		&
    		\hspace{7cm}\cdot\int_{ \mathcal{O}}\frac{e^{\rho(z_3)}}{z_3-z_2}Q_{21}(z_3)\,d\sigma_{z_3}\,d\sigma_{z_2}\,d\sigma_{z_1}\Bigg|^p\,d\sigma_{\lambda}\Bigg]^{1/p} \\
    		&\leq \sup_{z \in \mathcal O} \Bigg[C ||Q||_{L^{\infty}}^3\int_{ \mathcal{O}}
    		\int_{\mathcal{O}}\frac{1}{|z_1-z|}\frac{1}{|z_2-z_1|}\frac{1}{|z_2-w|^{1-\delta}} \, d\sigma_{z_2}\,d\sigma_{z_1}\Bigg] < C'
    		\end{align*}
    		
    		Thus, the result $(\ref{1007Q})$ follows, and in consequence also (\ref{0907F}) holds from (\ref{1107A}) and (\ref{1007P}).

    	\end{proof}
    	
    	\subsection{Scattering data and reconstruction of the potential in admissible points}
    	Let $w\in\mathcal{O}$ be an admissible point. We consider the function
    	\begin{equation}\label{0907C}
    	e^{\overline{\ln |\lambda| \lambda_s(z-w)^2}},
    	\end{equation}
    	where the number $\lambda_s$ is chosen as
    	\begin{equation}\label{0907D}
    	\sup_{z \in \overline{\mathcal O}} \textnormal{Re}[ \lambda_s(z-w)^2]<1/2, \quad \sup_{z \in \overline{\mathcal D}} \textnormal{Re}[ \lambda_s(z-w)^2]<-1/2.
    	\end{equation}
     
    	A point $w$ can be admissible to more than one spectral parameter. To define the scattering data we want to use the above exponentials depending on $\lambda_s$. Since $\mu$ fulfills an asymptotic with the spectral parameter being $\lambda_{\mathcal O}$ we also define our scattering data with respect to this spectral parameter, i.e., $\lambda_s=\lambda_{\mathcal O}$.
    	However, this is not a requirement and $\lambda_s$ could have been one of the other parameters which makes $w$ admissible. We just fix it like this to simplify the proofs ahead.
    	
    	Consider now our scattering data
    	\begin{equation}\label{1007S}
    	h(\lambda,w)= \int_{\partial \mathcal O} \overline{e^{\ln |\lambda| \lambda_s(z-w)^2}} \mu_2(z)d\bar{z}.
    	\end{equation}
    	Using Green's theorem
    	$$
    	\int_{\partial \mathcal O} f \, d\bar{z} = -2i \int_{\mathcal O} \partial {f} \, d{\sigma_{z}}
    	$$
    	we can see that
    	\begin{equation}\label{0907E}
    	h(\lambda,w)= \int_{\Gamma^+} \overline{e^{\ln |\lambda| \lambda_s(z-w)^2}} \mu_2(z)d\bar{z} + \int_{\mathcal O \setminus\overline{\mathcal{D}}}
    	\overline{e^{\ln |\lambda| \lambda_s(z-w)^2}} e^{-i\,\textnormal{Im} [\lambda (z-w)^2]/2} q_{21}(z)\mu_1(z) d\sigma_z.
    	\end{equation}
    	
    	This formula gives raise to an operator that we denote by $T$ and  it is defined by
    	
    	$$
    	T[G](\lambda)= \int_{\mathcal O \setminus\overline{\mathcal{D}}}
    	\overline{e^{\ln |\lambda| \lambda_s(z-w)^2}} e^{-i\,\textnormal{Im} [\lambda (z-w)^2]/2} q_{21}(z)G(z) d\sigma_z.
    	$$
    	
    	From our representation for the solution $\mu$ $(\ref{1107A})$ and the fact that the matrix $\widetilde{Q}_{\lambda}$ is off-diagonal we get
    	
    	$$
    	T\left [ \left ( (I+D\widetilde{Q}_\lambda)\left(
    	\begin{array}{c}
    	U \\
    	0 \\
    	\end{array}
    	\right) \right )_1\right ] = T[U].
    	$$
    	This allows us to state our main theorem.
    	
    	\begin{theorem}\label{210918T} 
    		Let the potential $q$ being given through (\ref{char1bc}) for an admittivity $\gamma$ satisfying $\textnormal{Re}(\gamma)\geq c>0$ and $\gamma\in W^{2,\infty}(\mathcal{D})\cap W^{2,\infty}(\mathcal{O}\setminus\overline{\mathcal{D}})$. If the jump $\alpha-1$ is small enough in $L^{\infty}(\Gamma)$,  and $w$ is a proper admissible point for a spectral parameter $\lambda_s$, then
    		\begin{equation}\label{210918A}
    		\frac{\lambda_s}{4\pi^2\ln 2}\lim_{R\to\infty}\int_{R<|\lambda|<2R} |\lambda|^{-1} \,  h(\lambda,w)\, d\sigma_\lambda=q_{21}(w).
    		\end{equation}
    	\end{theorem}
    	
    	The proof of this theorem requires some additional results concerning the behavior of $h(\lambda,w)/|\lambda|.$ These results will be given in the form of three lemmas which we establish in the next section.

    	\subsection{Necessary results for the proof of Theorem~\ref{210918T}}
    	
    	We start by presenting a result which we need afterwards. For its proof we refer to Appendix A.
    	Consider two arbitrary numbers $\lambda_0,\,w \in \mathbb C$, denote
    	$\rho(z)= -i\,\textnormal{Im}[\lambda(z-w)^2]/2+\ln |\lambda| \lambda_0(z-w)^2$, and let
    	$A_0=\sup_{z \in \mathcal O}\, \textnormal{Re}[ \lambda_0(z-w)^2].$ 
    	\begin{lemma}
    		\label{2310F}
    		Let $z_1 \in \mathbb C\setminus\{w\}$, $p > 2,$ and $\varphi \in L^\infty$ with compact support. Then
    		$$
    		\left \|\frac{1}{|\lambda|^{A_0}}\int_{\mathbb C} \varphi(z)\frac{e^{\rho(z)} }{z-z_1} \, d{\sigma_{z}} \right \|_{L^p_\lambda(\mathbb C)}
    		\leq C\frac{\|\varphi\|_{L^\infty}}{|z_1-w|^{1-\delta}},
    		$$
    		where the constant $C$ depends only on the support of $\varphi$ and on $\delta =\delta(p)>0$.
    	\end{lemma}
    	
    	To study the main term in (\ref{210918A}), we have the following lemma.
    	
    	\begin{lemma}\label{lemma08007A}
    		Let $\varphi \in W^{1,\infty}(\Omega)$, and $\rm{supp}(\varphi) \subset \Omega$, where $\Omega$ is a domain in $\mathbb R^2$ and $w \in \Omega$ such that
    		\begin{equation}\label{0807C}
   				 \sup_{z \in \overline{\Omega}} \textnormal{Re} (z-w)^2 <1. 
			\end{equation}
    		Then the following asymptotic holds
    		\begin{equation}\label{0807B}
    		\int_{\Omega} e^{-i\,\textnormal{Im} (\lambda (z-w)^2)+\ln|\lambda|(z-w)^2} \varphi(z) d\sigma_z = \frac{2\pi}{|\lambda|}\varphi(w)+R_w(\lambda),
    		\end{equation}
    		where $|\lambda|^{-1}R_w \in L^1(\lambda :|\lambda|>R)$.
    	\end{lemma}
    	
    	\begin{proof} Consider two domains 
    		$$
    		I_1=\{ z\in\Omega  ~ : ~ |z-w| < 3\varepsilon\} \quad\text{and} \quad I_2=\{ z\in\Omega  ~ : ~ |z-w| > \varepsilon \}, \quad
    		$$
    		where $\varepsilon>0$ is an a priori chosen arbitrarily small but  fixed number.  Furthermore,  we pick two functions $\delta_1$ and $\delta_2$ with supports $I_1$ and $I_2$, respectively,  such that $\delta_1+\delta_2 \equiv 1$ in $\mathcal O$. Moreover, we assume that $\delta_1(z-w)$ is represented as a product of $\widehat{\delta}_1(x)\widehat{\delta}_1(y)$ and that the function $\widehat{\delta}_1(x)$ decreases monotonically  as $|x|$ grows. 
    		
    		The integrand is multiplied by $(\delta_1+\delta_2)$ and this naturally splits the integral into two terms. 
    		The term corresponding to $\delta_2$ can be integrated by parts once and then the required estimate follows from the Hausdorff-Young inequality (\ref{0607F}) for $p=q=2$.  We also use here the fact that the inequality (\ref{0807C}) is strict.
    		
    		Now, we consider the term corresponding to integration against $\delta_1$. 
    		This term will be divided into two parts as well correspondent to the representation
    		$$
    		\delta_1(z) \varphi(z)=\delta_1(z)\varphi(w)+ \delta_1(z)(\varphi(z)-\varphi(w))
    		$$
    		Keeping in mind the properties of $\delta_1$ and the fact that $\varphi \in W^{1,\infty}$ is a Lipschitz function the second part can be treated as in Lemma \ref{2310F}, i.e., we make a change of variables $u=(z-w)^2$ and thus $d\sigma_{z}=\frac{1}{4|u|}d\sigma_{u}$. From this we obtain two integrals due to the splitting of the domain. In each one of them the change of variables generates a singularity of total order $|u|$. In both integrals we can apply integration by parts. We will obtain an area integral and a contour integral. On the area integral we have a singularity of total order $|u|^{3/2}$. Hence, we can apply Hausdorff-Young inequality for the Laplace transform for $p=4/3$ and obtain here the required estimate. For the contour integral we can apply the one-dimensional Hausdorff-inequality to the Laplace transform and obtain the needed estimate. 
    		
    		To the first part we consider the change of variables $y={\sqrt{|\lambda|}}(z-w)$.  Due to the separation of variables in $\delta_1$ the asymptotic of 
    		\begin{equation}\label{0807d}
    		\frac{\varphi(w)}{|\lambda|}
    		\int e^{-i\,\textnormal{Im} y^2+\frac{\ln|\lambda|}{|\lambda|}y^2} \delta_1\left(\left|w+\frac{y}{\sqrt{|\lambda|}}\right|\right) d\sigma_y
    		\end{equation}
    		follows from the formula
    		\begin{equation}\label{2207A}
    		\int_0^{\lambda^{1/2}\delta} e^{-ix^2+\frac{\ln|\lambda|}{|\lambda|}x^2} \widehat{\delta}_1\left(\frac{|x|}{\sqrt{|\lambda|}}\right) dx = \frac{1}{\sqrt{2\pi}}(1+o(1)), \quad \lambda \rightarrow \infty.
    		\end{equation}
    		
    		This can be proven in the following way:
    		consider the change of variables $$x^2=t, \;\,g(t):=\widehat{\delta}_1\left(\frac{|x(t)|}{\sqrt{|\lambda|}}\right)$$ then, we have
    		$$
    		\int_0^{\lambda^{1}\delta^2} e^{-it+\frac{\ln|\lambda|}{|\lambda|}t} \frac{1}{\sqrt t}g(t) dt =
    		$$
    		$$
    		\int_0^1 e^{-it+\frac{\ln|\lambda|}{|\lambda|}t} \frac{1}{\sqrt t}g(t) dt +
    		\frac{1}{-i+\frac{\ln|\lambda|}{|\lambda|}}\int_1^{\lambda \delta^2} \left( e^{-it+\frac{\ln|\lambda|}{|\lambda|}t} \right )' \frac{1}{\sqrt t}g(t) dt. 
    		$$
    		For the second term we obtain
    		$$
    		\frac{1}{-i+\frac{\ln|\lambda|}{|\lambda|}}\int_1^{\lambda \delta^2} \left( e^{-it+\frac{\ln|\lambda|}{|\lambda|}t} \right )' \frac{1}{\sqrt t}g(t) dt = 
    		\left . \frac{1}{-i+\frac{\ln|\lambda|}{|\lambda|}}\left( e^{-it+\frac{\ln|\lambda|}{|\lambda|}t} \right ) \frac{g(t)}{\sqrt t} \right |_1^{\lambda\delta^2}+
    		$$
    		$$
    		\frac{1}{-i+\frac{\ln|\lambda|}{|\lambda|}}\int_1^{\lambda \delta^2} e^{-it+\frac{\ln|\lambda|}{|\lambda|}t}  \frac{1}{2t^{3/2}}g(t) dt =
    		$$
    		$$
    		\frac{-1}{-i}\left( e^{-i} \right ) \frac{g(1)}{\sqrt 1} +\frac{1}{-i}\int_1^{\lambda \delta^2} e^{-it}  \left ( \frac{g(t)}{2t^{1/2}} \right )' dt +o(1), \quad \lambda \to \infty.
    		$$
    		We used here the fact that the last integral is absolutely convergent ($g$ has a finite support) and 
    		\begin{equation}\label{2207D}
    		\sup_{z \in I_1} \left |e^{\frac{\ln|\lambda|}{|\lambda|}y^2} -1 \right | = o(1), \quad \lambda \rightarrow \infty.
    		\end{equation}
    		
    		Therefore, we get  
    		$$
    		\int_0^{\lambda \delta^2} e^{-ix^2+\frac{\ln|\lambda|}{|\lambda|}x^2} f(x) dx = 
    		\int_0^{\lambda \delta^2} e^{-it} \frac{1}{\sqrt t}g(t) dt = \int_0^{\infty} e^{-it} \frac{1}{\sqrt t}g(t) dt +o(1).
    		$$
    		Now the result of our lemma is an immediate consequence of this formula.
    	\end{proof}
    	
    	To prove our asymptotic formula of the scattering data, we will substitute $\mu$ with the help of (\ref{1107A}). This will leave some terms which need to vanish as $|\lambda|\rightarrow +\infty$ in order to obtain the desired formula through Lemma \ref{lemma08007A}. 
    	In this sense, the two lemmas that follow assure this remaining terms are integrable in $\lambda$ and, therefore, their impact vanishes as we take the limit.
    	
    	\begin{lemma}\label{lemma1008A}
    		For some $p<2$, with $R$ large enough and  $f$ defined as in (\ref{1107A}), we get
    		\begin{equation}\label{1007A}
    		\frac{1}{|\lambda|}T\left [ M(I+D\widetilde{Q}_\lambda)\left(
    		\begin{array}{c}
    		U \\
    		0 \\
    		\end{array}
    		\right) \right ] \in L^p(\lambda : |\lambda|>R), \quad
    		\end{equation}
    		and
    		\begin{equation}\label{1107B}
    		\frac{1}{|\lambda|}T[Mf] \in L^p(\lambda : |\lambda|>R).
    		\end{equation}
    	\end{lemma}

    	\begin{proof}
    		Given the structure of $M=P\tilde{A}_{\lambda}+D\tilde{Q}_{\lambda}-D\tilde{Q}_{\lambda}D\tilde{Q}_{\lambda}$ and that $\frac{1}{|\lambda|}T$ is a linear operator, it is enough to show that each term applied to both, $\U$ and $D\tilde{Q}_{\lambda}\U$, belongs to $L^p(\lambda:|\lambda|> R)$.
    		
    		We look directly at the computations of each term. By using Fubini's Theorem, Minkowski integral inequality, H\"older inequality, and Lemma \ref{2310F} we can show that all of these terms are in fact in $L^p(\lambda:|\lambda|> R)$. Since the computations for each term follow roughly the same lines, and for the convenience of the reader, we present just the computation in one of these cases, the computations of the remaining terms being analogous, with special attention to the convergence of the integrals.
    		
    		We look at the term $$\frac{1}{|\lambda|}T\Bigg[D\tilde{Q}_{\lambda}D\tilde{Q}_{\lambda}\U\Bigg] \in L^p(\lambda:|\lambda|> R).$$
    		
    		Let us denote $\rho(z)=i\,\textnormal{Im}[\lambda(z-w)^2]/2+\ln|\lambda|\lambda_s(z-w)^2$ 
    		and $A=S=\sup_{z \in \overline{\mathcal{O}}} \,\textnormal{Re}[\lambda(z-w)^2] < 1/2$.
    		\begin{footnotesize}
    			
    			\begin{align*}
    			&\left\| \frac{1}{|\lambda|}T\Bigg[D\tilde{Q}_{\lambda}D\tilde{Q}_{\lambda}\U\Bigg] \right\|_{L^p(\lambda:|\lambda|> R)}=
    			\\
    			&=\Bigg[\int_{|\lambda| >  R} \Bigg| \frac{1}{4\pi^2|\lambda|}\int_{\mathcal{O}\setminus \overline{\mathcal{D}}}e^{\overline{\rho(z)}}q_{21}(z)\int_{\mathcal{O}}\frac{e^{-i\,\textnormal{Im}[\lambda(z_1-w)^2]/2}}{z_1-z}Q_{12}(z_1)
    			\int_{\mathcal{O}}\frac{e^{\rho(z_2)}}{\overline{z_2-z_1}}Q_{21}(z_2)\,d\sigma_{z_2}\,d\sigma_{z_1}\,d\sigma_z\Bigg|^p\,d\sigma_{\lambda}\Bigg]^{1/p}
    			\\
    			&= \Bigg[\int_{|\lambda| >  R} \Bigg|\frac{1}{4\pi^2|\lambda|}\int_{\mathcal{O}}\Bigg(\int_{\mathcal{O}\setminus\overline{\mathcal{D}}}\frac{e^{\overline{\rho(z)}}}{z_1-z}q_{21}(z)\,d\sigma_z\Bigg)\Bigg(\int_{\mathcal{O}}\frac{e^{\rho(z_2)}}{\overline{z_2-z_1}}Q_{21}(z_2)\,d\sigma_{z_2}\Bigg)\,\cdot
    			\\
    			&\hspace{2cm}\cdot Q_{12}(z_1)e^{-i\,\textnormal{Im}[\lambda(z_1-w)^2]/2} \, d\sigma_{z_1} \Bigg|^p \, d\sigma_{\lambda}\Bigg]^{1/p}
    			\\
    			&\leq \int_{\mathcal{O}} \Bigg[ \int_{|\lambda| >  R} \Bigg|\frac{|\lambda|^{A+S}}{|\lambda|}\Bigg(\frac{1}{|\lambda|^A}\int_{\mathcal{O}\setminus D}\frac{e^{\overline{\rho(z)}}}{z_1-z}q_{21}(z)\,d\sigma_z\Bigg) \Bigg(\frac{1}{|\lambda|^S}\int_{\mathcal{O}}\frac{e^{\rho(z_2)}}{\overline{z_2-z_1}}Q_{21}(z_2)\,d\sigma_{z_2}\Bigg)\Bigg|^p\, d\sigma_{\lambda}\Bigg]^{1/p}\hspace{-0.5cm}|Q_{12}(z_1)|\,d\sigma_{z_1}
    			\\
    			& \leq \|Q \|_{L^{\infty}} \int_{\mathcal{O}} \Bigg|\Bigg|\frac{1}{|\lambda|^A}\int_{\mathcal{O}\setminus \overline{\mathcal{D}}}\frac{e^{\overline{\rho(z)}}}{z_1-z}q_{21}(z)\,d\sigma_z \Bigg|\Bigg|_{L^{2p}(\lambda:|\lambda|>R)} \Bigg|\Bigg|\frac{1}{|\lambda|^S}\int_{\mathcal{O}}\frac{e^{\rho(z_2)}}{\overline{z_2-z_1}}Q_{21}(z_2)\,d\sigma_{z_2}\Bigg|\Bigg|_{L^{2p}(\lambda:|\lambda|>R)} \, d\sigma_{z_1}
    			\\
    			& \leq C \|Q \|_{L^{\infty}} \int_{\mathcal{O}} \frac{1}{|z_1-w|^{1-\delta}}\frac{1}{|z_1-w|^{1-\delta}} \, d\sigma_{z_1} < \infty.
    			\end{align*}
    		\end{footnotesize}
    		
    		With these calculations we obtain (\ref{1007A}). To show (\ref{1107B}) we have that $\frac{1}{|\lambda|^A}f \in \mathcal{H}^p$, for $p>2$, by Lemma \ref{lemma1107A}. We consider $T$ applied to each term of $M$. Again, we present only the computations for the case $\frac{1}{|\lambda|}T[D\tilde{Q}_{\lambda}D\tilde{Q}_{\lambda}f]$, since the other computations are analogous, with special attention to the behavior of $\frac{1}{|\lambda|^A}f$. In the same spirit, we only present the calculation for the first term of the vector.

    		\begin{align*}	
    		&\Bigg[\int_{|\lambda| >  R} \Bigg|\frac{1}{|\lambda|}\int_{\mathcal{O}\setminus\overline{\mathcal{D}}} e^{\overline{\rho(z)}}Q_{21}(z) \int_{ \mathcal{O}} \frac{e^{-i\,\textnormal{Im}(\lambda(z_1-w)^2)}}{z_1-z}Q_{12}(z_1)\cdot
    		\\
    		& \hspace{2cm}\cdot \int_{ \mathcal{O}}\frac{e^{i\,\textnormal{Im}(\lambda(z-w)^2)}}{\overline{z_2-z_1}}Q_{21}(z_2)f_1(z_2)\,d\sigma_{z_2}\,d\sigma_{z_1}\,d\sigma_{z}\Bigg|^p\,d\sigma_{\lambda}\Bigg]^{1/p}
    		\\
    		&
    		\leq C\|Q \|_{L^{\infty}}^3\int_{ \mathcal{O}}\int_{ \mathcal{O}} \frac{1}{|z_2-z_1|} \frac{1}{|z_1-w|^{1-\delta}} \left\|\frac{1}{|\lambda|^{A}}f_1(z_2)\right\|_{L^{2p}_{\lambda}}\,d\sigma_{z_2}\,d\sigma_{z_1} <  \infty 
    		\end{align*}

    		The boundedness of the last integral follows from the fact that $\frac{1}{|\lambda|^A}f \in \mathcal{H}^p$ implies its boundedness with respect to the $z$ variable.
    		
    	\end{proof}
    	
    	\begin{lemma}\label{260918A} 
    		For R large enough, and $w$ being a proper admissible point, we have
    		$$
    		\frac{1}{|\lambda|} \int_{\Gamma^+} e^{\overline{\ln |\lambda| \lambda_s(z-w)^2}} \mu_2(z)d\bar{z} \in L^1(|\lambda|>R).
    		$$
    	\end{lemma}
    	\begin{proof}
    		We divide the integral 
    		\begin{equation}	\label{121A}	
    		\frac{1}{|\lambda|}\int_{\Gamma^{+}} e^{\overline{\ln |\lambda|\lambda_s(z-w)^2}}\mu_2(z) d\bar{z},
    		\end{equation}
    		into two pieces, according to the decomposition of $\mu_2$ given by formula (\ref{1107A}), that is 
    		\begin{equation}\label{122A}
    		\mu_2=\bigg[D\tilde{Q}_{\lambda}\begin{pmatrix}
    		U \\
    		0 
    		\end{pmatrix}\bigg]_2 + f_2.
    		\end{equation}
    		
    		By Lemma \ref{lemma1107A} we have that $\frac{1}{|\lambda|^A} f \in \mathcal{H}^p$, for any $p > 2$. Therefore, we apply (\ref{122A}) to (\ref{121A}) and we split the integral into $I_1$ and $I_2$, according to the order in (\ref{122A}).
    		
    		Since, by assumption, $w$ is an admissible point there exists a $\lambda_s$ fulfilling the inequality  $\sup_{z \in \overline{\mathcal{D}}}\, \textnormal{Re}[\lambda_s(z-w)^2] < -1/2.$
    		
    		So, for $z \in \Gamma^{+}$ we get 
    		\begin{align}
    		\nonumber \left| |\lambda|^A e^{\overline{\ln |\lambda|\lambda_s(z-w)^2}} \right|&=|\lambda|^A|e^{\ln |\lambda|\,\textnormal{Re}[\lambda_s(z-w)^2]}|  < |\lambda|^A e^{-1/2\ln|\lambda|}=|\lambda|^{A-1/2} \\ 
    		\left| |\lambda|^A e^{\overline{\ln |\lambda|\lambda_s(z-w)^2}}\right| & <  |\lambda|^{-\delta},
    		\end{align}
    		where we choose $-\delta=A-1/2  <  0$ (recall, $A  <  1/2$). Hence, we obtain
    		\begin{align*}
    		|I_2| &\leq \frac{1}{|\lambda|}\int_{\Gamma^{+}} \bigg||\lambda|^Ae^{\overline{\ln |\lambda|\lambda_s(z-w)^2}} \bigg(\frac{1}{|\lambda|^A}f_2\bigg)\bigg| \, d|\bar{z}| \\
    		&  < \frac{1}{|\lambda|^{1+\delta}} \int_{\Gamma^{+}} \bigg|\frac{1}{|\lambda|^A}f_2\bigg| \, d|\bar{z}|.
    		\end{align*}
    		
    		Integrating with respect to the spectral parameter, we have for $R >  0$ large enough
    		\begin{align*}
    		\int_{|\lambda| >  R} |I_2| d\sigma_{\lambda} &\leq \int_{|\lambda| >  R} \frac{1}{|\lambda|^{1+\delta}} \int_{\Gamma^{+}} \bigg|\frac{1}{|\lambda|^A}f_2\bigg| \, d|\bar{z}| d\sigma_{\lambda} \\
    		& \leq \left\|\frac{1}{|\lambda|^{1+\delta}}\right\|_{L^q_{\lambda}} \left\|\int_{\Gamma^{+}} \bigg|\frac{1}{|\lambda|^A}f_2\bigg| \, d|\bar{z}|\right\|_{L^p_{\lambda}}.
    		\end{align*}
    		
    		Therefore, by Lemma \ref{lemma1107A} the second norm is finite for $p  > 2$. We now pick $q$  such that $q(1+\delta) > 2$, which is always possible given that $\delta  >  0$. 	Hence, $I_2$ is in $L^1(\lambda:|\lambda| > R).$
    		
    		Now, we look at $I_1$. By definition we have 
    		\begin{equation}
    		I_1=\frac{1}{2|\lambda|} \int_{\Gamma^{+}} e^{\overline{\ln |\lambda|\lambda_s(z-w)^2}}\int_{\mathcal{O}} \frac{e^{\ln |\lambda|\lambda_s(z_1-w)^2+i\,\textnormal{Im}(\lambda(z_1-w)^2)/2}}{\bar{z}-\bar{z}_1}Q_{21}(z_1) \,d\sigma_{z_1}\,d\bar{z}.
    		\end{equation}
    		
    		Again, integrating against the spectral parameter we get:
    		\small
    		\begin{align*}
    		\int_{|\lambda| >  R} |I_1| d\sigma_{\lambda} &\leq \int_{|\lambda| >  R} \frac{1}{2|\lambda|}\int_{\Gamma^{+}} |\lambda|^{-\delta} \Bigg|\frac{1}{|\lambda|^A}\int_{\mathcal{O}} \frac{e^{\ln |\lambda|\lambda_s(z_1-w)^2+i\,\textnormal{Im}(\lambda(z_1-w)^2)/2}}{\bar{z}-\bar{z}_1}Q_{21}(z_1) d\sigma_{z_1}\Bigg|\, d|\bar{z}|\,d\sigma_{\lambda} \\ 
    		& = \int_{\Gamma^{+}}\int_{|\lambda| >  R} \frac{1}{2|\lambda|^{1+\delta}} \Bigg|\frac{1}{|\lambda|^A}\int_{\mathcal{O}} \frac{e^{\ln |\lambda|\lambda_s(z_1-w)^2+i\,\textnormal{Im}(\lambda(z_1-w)^2)/2}}{\bar{z}-\bar{z}_1}Q_{21}(z_1) d\sigma_{z_1}\Bigg|\, d\sigma_{\lambda}\,d|\bar{z} \\
    		& \leq \Bigg|\Bigg|\frac{1}{2|\lambda|^{1+\delta}}\Bigg|\Bigg|_{L^q_{\lambda}} \int_{\Gamma^{+}} \Bigg|\Bigg|\frac{1}{|\lambda|^A}\int_{\mathcal{O}} \frac{e^{\ln |\lambda|\lambda_s(z_1-w)^2+i\,\textnormal{Im}(\lambda(z_1-w)^2)/2}}{\bar{z}-\bar{z}_1}Q_{21}(z_1) d\sigma_{z_1}\Bigg|\Bigg|_{L^p_{\lambda}} d|\bar{z}|,
    		\end{align*}
    		\normalsize
    		where we use Fubini's theorem and H\"older inequality, with $p > 2$ small enough so that the first norm is finite as in the computation of $I_2$.
    		
    		Now, we can use Lemma \ref{2310F}, given that we assume that our potential $Q$ has support in $O$ and it is in $L^{\infty}_{z}$, to obtain a constant $C>0$ depending only on the support of the potential and  on a certain $\tilde{\delta} > 0$:
    		\begin{align*}
    		\int_{|\lambda| >  R} |I_1| d\sigma_{\lambda} \leq C \left\| \frac{1}{2|\lambda|^{1+\delta}}\right\|_{L^q_{\lambda}}  \|Q_{21} \|_{L^{\infty}_{z}} \int_{\Gamma^{+}} \frac{1}{|\bar{z}-w|^{1-\tilde{\delta}}} \, d|\bar{z}|.
    		\end{align*}
    		
    		Given that the last integral is finite, we have $I_1 \in L^1(\lambda:|\lambda| > R)$ and the desired result follows. \end{proof}

    	\subsection{Proof of Theorem \ref{210918T} }
    	
    	Now we can present the proof of our main theorem, using the lemmas of the previous section while paying close attention to how $\mu$ and $f$ are defined.
    	
    	\begin{proof}
    		Let us start by taking a look at the following term
    		
    		\begin{align}\label{Hsc}
    		\nonumber\frac{h(\lambda,w)}{|\lambda|}= \frac{1}{|\lambda|}\Bigg[&\int_{\Gamma^{+}} e^{\overline{\ln|\lambda|\lambda_s(z-w)^2}}\mu_2(z)d\bar{z}\\
    		&+\int_{\mathcal{O}\setminus\overline{\mathcal{D}}} e^{\overline{\text{ln}|\lambda|\lambda_s(z-w)^2}}e^{-i\,\textnormal{Im}(\lambda(z-w)^2)/2} q_{21}(z)\mu_1(z) d\sigma_z\Bigg].
    		\end{align}

    		From (\ref{1107A}) we have
    		$$\mu=f+(I+D\tilde{Q}_{\lambda}) \begin{pmatrix}
    		U \\
    		0 
    		\end{pmatrix},$$
    		whereby $f$ is a solution of
    		$$f=-\Bigg(Mf + M(I+D\tilde{Q}_{\lambda}) \begin{pmatrix}
    		U \\
    		0 
    		\end{pmatrix}\Bigg).$$ 
    		
    		This leads to
    		$\mu_1=-\left[Mf+ M(I+D\tilde{Q}_{\lambda}) \begin{pmatrix}
    		U \\
    		0 
    		\end{pmatrix}\right]_1+U.$ 	Therefore, by (\ref{Hsc}) and the definition of the operator $T$, we get
    		\begin{align}
    		\nonumber \frac{h(\lambda,w)}{|\lambda|}=\frac{1}{|\lambda|}&\int_{\Gamma^{+}} e^{\overline{\ln|\lambda|\lambda_s(z-w)^2}}\mu_2(z)d\bar{z}-\frac{1}{|\lambda|}T\bigg(\big[Mf\big]_1\bigg)\\&-\frac{1}{|\lambda|}T\bigg( \bigg[M(I+D\tilde{Q}_{\lambda}) \begin{pmatrix}
    		U \\
    		0 
    		\end{pmatrix}\bigg]_1\bigg) +
    		\frac{1}{|\lambda|}T[U] =:A+B+C+D.		
    		\end{align}
    		We need to study the terms $A,\,B,\,C,\,D$.
    		By Lemma \ref{lemma1008A}, we have for $p < 2$ and $R$ large enough that:
    		$$B,C \in L^p(\lambda:|\lambda|> R).$$
    		
    		From Lemma \ref{260918A}, we obtain
    		$$A \in  L^1(\lambda:|\lambda|> R).$$
    		
    		Hence, we just need to analyze the behavior of the last term.
    		
    		\begin{align*}
    		T[U]&=\int_{\mathcal{O}\setminus\overline{\mathcal{D}}} e^{\overline{\text{ln}|\lambda|\lambda_s(z-w)^2}}e^{-i\,\textnormal{Im}(\lambda(z-w)^2)/2} q_{21}(z)e^{\ln|\lambda|\lambda_s(z-w)^2} d\sigma_z \\
    		&= \int_{\mathcal{O}\setminus\overline{\mathcal{D}}} e^{\overline{\text{ln}|\lambda|(\sqrt{\lambda_s}z-\sqrt{\lambda_s}w)^2}}e^{-i\,\textnormal{Im}(\lambda(z-w)^2)/2} q_{21}(z)e^{\ln|\lambda|(\sqrt{\lambda_s}z-\sqrt{\lambda_s}w)^2} d\sigma_z \\
    		&= \frac{1}{\lambda_s}\int_{\mathcal{O}\setminus \overline{\mathcal{D}}} e^{\overline{\text{ln}|\lambda|(z-\sqrt{\lambda_s}w)^2}}e^{-i\,\textnormal{Im}(\lambda/\sqrt{\lambda_s}(z-\sqrt{\lambda_s}w)^2)/2} q_{21}(z)\\ &    \hspace{2.5cm}e^{\ln|\lambda|(\sqrt{\lambda_s}z-\sqrt{\lambda_s}w)^2}e^{-i\,\textnormal{Im}(\lambda(z-\sqrt{\lambda_s}w)^2)} e^{i\,\textnormal{Im}(\lambda(z-\sqrt{\lambda_s}w)^2)} d\sigma_z,
    		\end{align*} 
    		where we did a simple change of variables.
    		We define
    		$$\phi(z)=e^{-i\,\textnormal{Im}(\lambda/\lambda_s(z-\sqrt{\lambda_s}w)^2)/2}e^{i\,\textnormal{Im}(\lambda(z-\sqrt{\lambda_s}w)^2)} e^{\overline{\text{ln}|\lambda|(z-\sqrt{\lambda_s}w)^2}} q_{21}(z/\sqrt{\lambda_s}).$$
    		
    		Given that the conditions of Lemma \ref{lemma08007A} are fulfilled, we obtain:
    		\begin{align}
    		T[U]&=\frac{1}{\lambda_s}\Bigg[\frac{2\pi}{|\lambda|}\phi(\sqrt{\lambda_s}w)+R_{\sqrt{\lambda_s}w}(\lambda)\Bigg],
    		\end{align}
    		which by substitution implies:
    		$$\frac{1}{|\lambda|}T[U]=\frac{1}{\lambda_s}\frac{2\pi}{|\lambda|}q_{21}(w)+\frac{1}{\lambda_s}|\lambda|^{-1}R_{\sqrt{\lambda_s}w}(\lambda)=:D_1+D_2$$
    		
    		By Lemma \ref{lemma08007A}, we have $D_2 \in L^1(\lambda:|\lambda|> R)$.
    		
    		So finally we are ready to evaluate the left-hand side of (\ref{210918A}):
    		\begin{align*}
    		\lim_{R \to \infty} \int_{R <  |\lambda|< 2R} |\lambda|^{-1} h(\lambda,w) d\sigma_{\lambda} &= \lim_{R \to \infty} \int_{R <  |\lambda|< 2R} \frac{2\pi}{\lambda_s}|\lambda|^{-2}q_{21}(w) d\sigma_{\lambda} \\
    		&= q_{21}(w)\frac{4\pi^2}{\lambda_s} \lim_{R \to \infty} \int_{R}^{2R} r^{-1} dr \\
    		&=  q_{21}(w)\frac{4\pi^2}{\lambda_s} \lim_{R \to \infty} \ln r\Big|_{R}^{2R} = q_{21}(w)\frac{4\pi^2\ln 2}{\lambda_s}.
    		\end{align*} 
    		
    		From this we get the desired asymptotic:
    		$$ q_{21}(w)=\frac{\lambda_s}{4\pi^2\ln 2} \lim_{R \to \infty} \int_{R <  |\lambda|< 2R} |\lambda|^{-1} h(\lambda,w) d\sigma_{\lambda}.$$
    		
    	\end{proof}
    	\section{Scattering data for Dirac equation via the Dirichlet-to-Neumann map} \label{ScDN}
    	
    	Our next goal is to establish a relation between the Dirichlet-to-Neumann map for equation (\ref{set27A}) and the traces of the solutions of (\ref{firbc}) on $\partial\mathcal O$. Let
    	
    	$$
    	\mathcal{T}_q:= \left\{\phi|_{\partial \mathcal O} : \quad \phi=\left(
    	\begin{array}{c}
    	\phi_1 \\
    	\phi_2 \\
    	\end{array}
    	\right)
    	\mbox{ is a solution of } (\ref{firbc}), \quad \phi_1,  \phi_2 \in H^{1}(\mathcal O)\right\}.
    	$$
    	
    	Let $u\in H^{2}( \mathcal O\setminus\overline{\mathcal{D}})\cap H^2(\mathcal{D})$ be a solution of (\ref{set27A}) with $u|_{\partial \mathcal O}= f \in H^{3/2}(\partial \mathcal O)$. Consider $\phi=\gamma^{1/2}(\partial u, \bar{\partial} u) \in  H^{1}(\mathcal O\setminus\overline{\mathcal{D}})\cap H^1(\mathcal{D})$. Then, formally
    	\begin{equation}\label{1112A}
    	\phi|_{\partial \mathcal O}= \frac{1}{2}
    	\left ( \!\begin{array}{cc} \ \bar{\nu} & -i\bar{\nu} \\ \nu  & i\nu \end{array} \!\right )
    	\left ( \!\!\!\begin{array}{c} \Lambda_\gamma  f  \\ \partial_s f \end{array} \!\!\!\right ),
    	\end{equation}
    	where $\Lambda_\gamma$ is the co-normal D-t-N map and $\partial_s$ is the operator of the tangential derivative. Inverting we get
    	\begin{equation}\label{1112C}
    	\left ( \!\!\!\begin{array}{c} \Lambda_\gamma  f  \\ \partial_s f \end{array} \!\!\!\right )=
    	\left ( \!\begin{array}{cc} {\nu} & \bar{\nu} \\ i\nu  & -i\bar{\nu} \end{array} \!\right ) \phi|_{\partial \mathcal O}.
    	\end{equation}
    	We normalize $\partial_s^{-1}$ in such a way that
    	$$
    	\int_{\partial \mathcal O} \partial_s^{-1}f ds =0.
    	$$
    	Then (\ref{1112C}) could be rewritten as {\it a boundary relation}
    	\begin{equation}\label{1112B}
    	(I-i\Lambda_\gamma\partial_s^{-1}) (\nu \phi_1|_{\partial \mathcal O}) = (I+ i\Lambda_\gamma\partial_s^{-1}) (\bar{\nu} \phi_2|_{\partial \mathcal O})
    	\end{equation}
    	Let us show the generalization of  \cite[Thm 3.2]{knud}, where $\gamma \in  C^{1+\epsilon}(\mathbb{R}^2)$, to the case of non-continuous $\gamma$.
    	
    	\begin{theorem}\label{Th1112A}
    		$$
    		\mathcal{T}_q = \left\{ (h_1,h_2)\in H^{1/2}(\partial \mathcal O) \times  H^{1/2}(\partial \mathcal O)  : (I-i\Lambda_\gamma\partial_s^{-1}) (\nu h_1) = (I+ i\Lambda_\gamma\partial_s^{-1}) (\bar{\nu} h_2) \right\}
    		$$
    	\end{theorem}
    	\begin{proof}
    		First we show that any pair $(h_1,h_2)^t\in H^{1/2}(\partial \mathcal O) \times  H^{1/2}(\partial \mathcal O)$ that satisfies the boundary relation above is in $\mathcal{T}_q$.
    		Consider a solution $u\in H^2(\mathcal O\setminus \overline{\mathcal{D}})\cap H^2(\mathcal{D})$ of (\ref{set27A}) with the boundary condition
    		$$
    		u|_{\partial \mathcal O} = i\partial_s^{-1}(\nu h_1 -\bar{\nu} h_2)\in H^{3/2}(\partial \mathcal O).
    		$$
    		Since $\gamma\in  W^{1,\infty}(\mathcal O\setminus\overline{\mathcal{D}})\cap W^{1,\infty}(\mathcal D)$ and $\gamma$ is separated from zero, it follows that $\gamma^{1/2}\in  W^{1,\infty}(\mathcal O\setminus\overline{\mathcal{D}})\cap W^{1,\infty}(\mathcal{D})$. Then, both components of the vector $\phi=\gamma^{1/2}(\partial u, \bar{\partial} u)^t$ belong to $H^{1}( \mathcal O\setminus\overline{\mathcal{D}})\cap H^1(\mathcal D)$ and $\phi$ satisfies (\ref{firbc}). The fact  $\phi|_{\partial \mathcal O}=(h_1,h_2)^t$ follows from (\ref{1112A}) and (\ref{1112B}).
    		
    		Conversely, we start  with a solution $\phi \in H^1(\mathcal O\setminus \overline{\mathcal{D}})\cap H^1(\mathcal{D})$ of (\ref{firbc}) satisfying (\ref{7JunB}) on $\Gamma$. From (\ref{firbc}) and (\ref{char1bc}) the following compatibility condition holds on $\Gamma$
    		$$
    		\bar{\partial} (\gamma^{-1/2} \phi_1) = {\partial} (\gamma^{-1/2} \phi_2).
    		$$
    		The Poincaré lemma ensure the existence of a function $u$ such that
    		$$
    		\left ( \!\!\!\begin{array}{c} \phi_1 \\ \phi_2 \end{array} \!\!\!\right ) =
    		\gamma^{1/2}\left ( \!\!\!\begin{array}{c} \partial u \\ \bar{\partial} u\end{array} \!\!\!\right ) \quad \text{on} \quad \mathcal{O}\setminus\Gamma.
    		$$
    		It is easy to check that $u$ is a solution to (\ref{set27A}) on $\mathcal{O\setminus}\Gamma$ and belongs to $H^2(\mathcal O\setminus\overline{\mathcal{D}})\cap H^2(\mathcal{D})$. Moreover, through the Poincaré Lemma and (\ref{7JunB}) it satisfies the transmission condition (\ref{TransCond}). Then, (\ref{1112A})-(\ref{1112B}) proves that $h=\phi|_{\partial \mathcal O}$ satisfies the boundary relation stated in the theorem.
    	\end{proof}
    	
    	Denote $S_{\lambda,w} : H^{1/2}(\partial \mathcal O) \rightarrow H^{1/2}(\partial \mathcal O) $
    	$$
    	S_{\lambda,w} f (z)=\frac{1}{i\pi} \int_{\partial \mathcal O}f(\varsigma) \frac{e^{-\lambda(z-w)^2+ \lambda(\varsigma -w)^2}}{\varsigma - z} d \varsigma
    	$$
    	This integral is understood in the sense of principal value.
    	
    	A future idea to explore is to determine conditions on how to find the trace of $\phi$ at $\partial\mathcal{O}$. A hint to this is given in \cite[Th. III.3.]{knud} for $\gamma \in C^{1+\varepsilon}(\mathbb R^2)$, although in here the exponential growing solutions are of the type $e^{ikz}$. 
    	
    	In this sense, we state a conjecture that we would like to prove, in future work, for our method:
    	\begin{conjecture}
    		The only pair $(h_1,h_2) \in H^{1/2}(\partial \mathcal O) \times H^{1/2}(\partial \mathcal O) $ which satisfies
    		\begin{eqnarray}\label{1611A}
    		(I-S_{\lambda,w} )h_1=2e^{\lambda (z-w)^2}, \\
    		(I-\overline{S_{\lambda,w}})h_2=0, \\
    		(I-i\Lambda_\gamma\partial_s^{-1}) (\nu h_1) = (I+ i\Lambda_\gamma\partial_s^{-1}) (\bar{\nu} h_2)
    		\end{eqnarray}
    		is $(\left.\phi_1\right|_{\partial\mathcal{O}},\left.\phi_2\right|_{\partial\mathcal{O}})$, where $\phi_1,\phi_2$ are the solutions of the Dirac equation (\ref{firbc}) satisfying the asymptotics (\ref{8JunA}).
    		
    	\end{conjecture}
    
    \section{Appendix A.}
    Here we show the proof of Lemma \ref{2310F}, which corresponds to the application of the Laplace Transform analogue of the Hausdorff-Young inequality. This lemma stems from conversations and discussions with  S. Sadov \cite{Sadov}. Our deep thanks.
    
    \section{Laplace Transform analogue of the Hausdorff-Young inequality}
    We need to recall some statements on the Laplace Transform. 
    
    The following results hold (see \cite{Sadov}): consider the map
    $$
    L_\gamma f(s)= \int_0^\infty e^{-z(s)t}f(t)dt
    $$
    where $s>0$ is the arclength of a contour $\gamma ~:~= \{z(s): s>0,\, \textnormal{Re}(z(s))>0 \}$.
    
    Theorem 7 from \cite{Sadov} claims that $L_\gamma$ is a bounded operator from $L^q(\mathbb R^+)$ to $L^p(\gamma)$, where $1 \leq q \leq 2$ and $1/p+1/q=1$.
    Moreover, the norm of this map is bounded uniformly in the class of convex contours.
    
    Now we consider only contours such that $|(\textnormal{Re} z(s))'|<1/2$ for $s>>1$. This means that the spaces  $L^p, p>1$  for the variable $s>0$ and for variable $\,\textnormal{Im} z(s)$ are equivalent.  We now prove that the result of the Hausdorff-Young inequality is valid for the following map on the plane:
    $$
    L f(\lambda_1,\lambda_2)= \int_0^\infty   \int_1^\infty   e^{-i\lambda_1 x} e^{-i\lambda_2 y -\ln|\lambda_2| y} f(x,y)d x d y, \quad \lambda_1,\lambda_2>0,
    $$
    namely we prove that for some fixed domain $D$ and constant $C=C(D)>0,$ we have
    \begin{equation}\label{0607F}
    \|Lf\|_{L^p_{\lambda_1,\lambda_2}} \leq C \|f\|_{L^q_{x,y}}, \quad \mbox{Supp} f \subset D.
    \end{equation}
    {\bf Proof}
    Consider the function $A(y,\lambda_1)$:
    \begin{equation}\label{0607A}
    A(y,\lambda_1)= \int_0^\infty   e^{-i\lambda_1 x} f(x,y)d x, \quad y>1
    \end{equation}
    and note that by Hausdorff-Young inequality we get 
    \begin{equation}\label{0607E}
    \|A(y,\cdot)\|_{L^p_{\lambda_1}} = \left ( \int |A(y,\lambda_1)|^p d\lambda_1 \right )^{1/p} \leq \left ( \int |f(x,y)|^qdx \right )^{1/q}.
    \end{equation}
    For the sake of simplicity we omit all positive constants here and in further inequalities.
    We claim that $A(\cdot,\lambda_1) \in L^q_y$ and we prove this fact later. Accepting this claim and using the above mentioned theorem from \cite{Sadov} we get 
    \begin{equation}\label{0607B}
    \int \left | \int e^{-i\lambda_2 y -\ln|\lambda_2| y} A(y,\lambda_1) dy \right |^p d\lambda_2 \leq (\|A(\cdot,\lambda_1)\|_{L^q_y})^p.
    \end{equation}
    Further we use the notation $\|\cdot \|$ for $\|\cdot \|_{L^p_{\lambda_1,\lambda_2}}$. Now we are ready to estimate $\|Lf\|$:
    \begin{equation}\label{0607C}
    \|Lf\|^p= \int \int \left | \int e^{-i\lambda_2 y -\ln|\lambda_2| y} A(y,\lambda_1) dy  \right |^p d\lambda_1 d\lambda_2 \leq
    \int (\|A(\cdot,\lambda_1)\|_{L^q_y})^p d \lambda_1.
    \end{equation}
    First we apply the integral form of the Minkowski inequality, and then (\ref{0607E}). Hence,  we get:
    \begin{equation}\label{0607D}
    \left ( \int \|A(\cdot,\lambda_1)\|^p_{L^q_y} d \lambda_1 \right )^{q/p}= \left ( \int \left | \int |A(y,\lambda_1)|^q dy \right |^{p/q} d\lambda_1 \right )^{q/p} \leq
    \end{equation}
    $$
    \int \left ( \int |A(y,\lambda_1)|^p d\lambda_1 \right )^{q/p} dy \leq \int \left ( \int |f(x,y)|^qdx \right ) dy.
    $$
    This proves (\ref{0607F}). Now let us show that $A(\cdot,\lambda_1) \in L^q_y$. From Minkowski inequality we get
    $$
    \|A(\cdot,\lambda_1)\|_{L^q_y} = \left (\int \left | \int_0^\infty e^{-i\lambda_1 x} f(x,y)dx  \right |^q dy \right)^{1/q} \leq
    \int \left ( \int |f(x,y)|^q dy \right )^{1/q} dx.
    $$
    Since $f$ has finite support, then the function $\int |f(x,y)|^q dy$ has  finite support too. Let us denote by $C_1$ the length of its support. Therefore
    $$
    \int \left ( \int |f(x,y)|^q dy \right )^{1/q} dx \leq \int \left ( \int |f(x,y)|^q dy \right )^{} dx + C_1.
    $$
    
    \section{Proof of Lemma \ref{2310F}}
    
    The following lemma represents a generalization of Lemma 3.2 from \cite{ltv}.
    Consider $\lambda_0 \in \mathbb C$, denote by 
    $\rho(z)= -i\,\textnormal{Im}[\lambda(z-w)^2]/2+\ln |\lambda| \lambda_0(z-w)^2$, and let
    $A_0=\sup_{z \in \mathcal O}\, \textnormal{Re}[ \lambda_0(z-w)^2].$ For convenience we recall Lemma \ref{2310F}:\\
    
    \textbf{Lemma}
    \textit{Let $z_1,w \in \mathbb C$, $p > 2$ and $\varphi \in L^\infty_{\rm{comp}}$. Then
    	$$
    	\left \|\frac{1}{|\lambda|^{A_0}}\int_{\mathbb C} \varphi(z)\frac{e^{\rho(z)} }{z-z_1} \, d{\sigma_{z}} \right \|_{L^p_\lambda(\mathbb C)}
    	\leq C\frac{\|\varphi\|_{L^\infty}}{|z_1-w|^{1-\delta}},
    	$$
    	where the constant $C$ depends only on the support of $\varphi$ and on $\delta =\delta(p)>0$.}
    
    \begin{proof}
    	Denote by $F=F(\lambda,w,z_1)$ the integral on the left-hand side of the inequality above. In order to have non-positiveness of the real part of the phase we make a change of  variables $u=(z-w)^2$ in $F$ and take into account that $d{\sigma_{u}} =4|z-w|^2 d{\sigma_{z}}$. Then
    	\begin{equation}\label{250918A}
    	F=\frac{1}{4}\sum_{\pm} \int_{\mathbb C} \varphi(w\pm\sqrt{u})\frac{e^{i\,\textnormal{Im}(\lambda u)/2+\lambda_0 \ln|\lambda| u }}{|u|(\pm\sqrt{u}-(z_1-w))} \, d{\sigma_{u}}.
    	\end{equation}
    	Now, we consider a new change of variable $\widehat{u}=u-u^0$, where
    	$$
    	u^0=\rm{argmax}_{w\pm\sqrt{u} \in \rm{supp} \varphi}\,\textnormal{Re}(\lambda_0 u)
    	$$
    	and apply the Hausdorff-Young inequality for the Laplace transform on a contour (\ref{0607F}). The result on Lemma \ref{2310F} follows immediately from \cite{ltv}, Lemma 3.1 which we recall here for the reader's convenience
    	
    	{\it \cite[Lemma 3.1]{ltv}
    		Let $1\leq p<2$. Then the following estimate is valid for an arbitrary $0 \neq a \in \mathbb C$ and some constants $C=C(p,R)$ and $\delta=\delta(p)>0$:
    		$$
    		\left \| \frac{1}{u(\sqrt{u}-a)} \right \|_{L^p(u\in \mathbb C:|u|<R)} \leq C(1+|a|^{-1+\delta}).
    		$$
    	}
    \end{proof}

\end{document}